\documentclass[11pt]{article}

\usepackage[a4paper,margin=1in]{geometry}
\usepackage{amsmath,amssymb,amsthm,mathtools}
\usepackage{microtype}
\usepackage{xcolor}
\usepackage{booktabs}
\usepackage{graphicx}
\usepackage{algorithm}
\usepackage{algpseudocode}
\usepackage{authblk}
\usepackage{tikz}
\usetikzlibrary{arrows.meta,calc,decorations.pathreplacing,positioning}
\usepackage[
colorlinks=true, linkcolor=blue!55!black, citecolor=blue!55!black, urlcolor=blue!55!black ]{hyperref}

\newtheorem{theorem}{Theorem}[section]
\newtheorem{lemma}[theorem]{Lemma}
\newtheorem{proposition}[theorem]{Proposition}
\newtheorem{corollary}[theorem]{Corollary}
\theoremstyle{definition}
\newtheorem{definition}[theorem]{Definition}

\newcommand{\ALG}{\operatorname{ALG}}
\newcommand{\OPT}{\operatorname{OPT}}
\newcommand{\cost}{\operatorname{cost}}
\newcommand{\E}{\mathbb{E}}
\newcommand{\Env}{\mathcal{E}}
\newcommand{\OPTTimer}{\operatorname{OPT\text{-}Timer}}

\title{Online Multi-Level Aggregation with Maximum Delay}
\author{
Tianhang Lu\thanks{The first two authors contribute equally.}, Runtian Ren, Shengcai Liu, Ke Tang } \affil{ Guangdong Provincial Key Laboratory of Brain-Inspired Intelligent Computation,\\ Department of Computer Science and Engineering,\\ Southern University of Science and Technology, Shenzhen 518055, China\\ liusc3@sustech.edu.cn }
\date{}

\begin{document}

\maketitle

\begin{abstract}
We study Multi-Level Aggregation with Maximum Delay (MLA-MD) on finite weighted rooted trees.
Each service pays for a rooted subtree and for the maximum waiting time among the requests it clears.
For the offline problem, we give an exact polynomial-time algorithm that computes both the optimum and an optimal service schedule.
For the online problem, we give a deterministic $2$-competitive algorithm and, against an oblivious adversary, a randomized $e/(e-1)$-competitive algorithm.
Both competitive ratios are optimal, even on rooted paths.
We also extend the offline result and both online guarantees to service systems with submodular costs.
\end{abstract}

\clearpage

\section{Introduction}
\label{sec:introduction}

Grouping requests can reduce service cost, but waiting to form a larger group also creates delay.
This tradeoff appears in joint replenishment~\cite{aksoy1988multi,goyal1989joint,joneja1990joint}, lot sizing and logistics~\cite{sindhuchao2005integrated,brahimi2006single,jans2008modeling,quadt2008capacitated,bushuev2015review,karimi2003capacitated}, and communication networks~\cite{yuan2003synchronization,leung2007overview,khanna2002control}.
Multi-level aggregation (MLA) models shared service costs on a weighted tree.
Requests arrive at its nodes, and each service pays for a connected subtree containing the root.
When their paths to the root overlap, requests share the cost of the common edges.
Because future arrivals are unknown, an online algorithm must decide when to serve and how far the service subtree should reach.

We study \emph{Multi-Level Aggregation with Maximum Delay} (MLA-MD).
Here, a service batch pays for the longest wait among its requests, rather than the sum of their waits.
At time $t$, a service selects a subtree $S$ and clears the nonempty batch $B$ of pending requests it covers.
If request $r$ arrived at time $a(r)$, the service costs $w(S)+\max_{r\in B}\{t-a(r)\}$ where $w(S)$ is the weight of $S$.
The objective is to minimize the sum of these service costs.
Thus the oldest request determines the batch’s waiting cost.
This objective is related to work on TCP delays with and without batch dependence~\cite{bhore2026online} and on objectives based on the longest wait~\cite{albers2005dynamic}.
Here, each batch also pays for the shared part of the tree used to serve it.

Maximum delay changes what makes a pending batch urgent.
Once the oldest request is fixed, adding younger requests leaves the batch's delay cost unchanged.
By contrast, every pending request contributes under additive delay.
Algorithms for additive MLA often organize urgency by tree levels or local scales~\cite{bienkowski2020mlap,bienkowski2021new,buchbinder2017depth,azar2019framework,mcmahan2021dcompetitive}.
However, those algorithms do not automatically work for a batch charged by its maximum delay.
Consider an adaptation of \emph{Balance}~\cite{bienkowski2013chain}: it keeps dyadic prefixes and lets the oldest pending request in each prefix trigger service.
We call this algorithm \emph{Max-Balance}.
In Appendix~\ref{app:balance-separation}, we show that even on line instances with $N$ requests, \emph{Balance} and \emph{Max-Balance} have competitive ratios of $\Omega(\sqrt N)$ and $\Omega(\log(N+1))$, respectively.
Thus replacing an additive local urgency measure with a maximum does not give a constant competitive ratio.

When all requests lie at a single positive-cost node, MLA-MD reduces to the maximum-delay TCP acknowledgment problem studied by Dooly, Goldman, and Scott~\cite{dooly2001tcp}.
In this fixed-cost setting, they gave an exact offline dynamic program and an optimal deterministic algorithm that restarts a timer after each arrival.
On a tree, however, the cost of a service depends on the root paths it covers, coupling the decision of when to serve with the choice of which requests to serve together.

\paragraph{Our results.}
We study MLA-MD on weighted rooted trees and present the following results.
\begin{enumerate}
\item We give a polynomial-time algorithm for computing an optimal offline schedule for MLA-MD on weighted rooted trees.

\item We give a deterministic online algorithm with competitive ratio $2$ and prove a matching lower bound.

\item Against an oblivious adversary, we give a randomized online algorithm with competitive ratio $e/(e-1)$ and prove a matching lower bound.

\item The offline algorithm and both online guarantees extend to the submodular service-cost model.
The competitive ratios remain optimal for this broader class.
\end{enumerate}

\paragraph{Our techniques.}
The offline algorithm is based on an interval dynamic program.
Its key structural ingredient is a temporal uncrossing argument.
Although a schedule may group arbitrary requests, batches whose arrival intervals overlap can be merged without increasing cost.
Thus some optimum partitions the arrival epochs into consecutive blocks, each served when its last epoch arrives, and the dynamic program optimizes over these blocks.

The deterministic algorithm runs this dynamic program on the currently revealed busy block and uses its offline value to set a service deadline.
A parameter controls how much of the joint-service cost is added to the timer.
An accepted arrival never advances the deadline, and increasing the parameter can only merge adjacent blocks.
The timer choices therefore form a nested family of partitions.
At the deterministic endpoint, the waiting cost of each block equals its offline value and its service cost is no larger.
The block values sum to the global offline optimum, giving the factor $2$.

The randomized algorithm samples one parameter before the input and retains it.
A potential compares the resulting nested partitions.
Between merges, its slope records spatial service cost.
When blocks merge, it can only jump upward.
Exponential weighting converts this comparison into the ratio $e/(e-1)$ against an oblivious adversary.

The tree enters only through the minimum cost of serving requests jointly.
On a rooted tree, this cost is nonnegative, normalized, monotone, and submodular.
These properties support the uncrossing and partition arguments, together with the condition that a minimum-cost service for the current pending set can be executed without affecting later services.
The offline dynamic program and both online analyses therefore extend to the submodular service-cost model.

\paragraph{Related work.}
\emph{Aggregation with additive delays.} The TCP acknowledgment problem is a special case in which each service incurs a fixed cost.
Its optimal competitive ratios are $2$ for deterministic algorithms~\cite{dooly2001tcp} and $e/(e-1)$ for randomized algorithms~\cite{seiden2000guessing,karlin2001dynamic}.
Joint replenishment introduces several item types with a shared ordering cost~\cite{buchbinder2008online,bienkowski2014better}.
Aggregation on a line was studied in~\cite{khanna2002control,brito2004competitive}, and \emph{Balance} gives a $5$-competitive algorithm for the additive delay objective~\cite{bienkowski2013chain}.
Multi-level aggregation places these problems in a tree hierarchy.
The depth-one and depth-two cases capture TCP acknowledgment and joint replenishment, respectively.
Deterministic algorithms for additive delays and deadlines have been developed through tree and general metric frameworks~\cite{bienkowski2020mlap,bienkowski2021new,buchbinder2017depth,azar2019framework,lin2026improved}.
McMahan's arXiv preprint claims a $D$-competitive algorithm for the deadline variant, but Lin and Turoczy report that they could not confirm its correctness~\cite{mcmahan2021dcompetitive,lin2026improved}.
Moreover, a stochastic version of MLA with additive delay under Poisson arrivals was studied in~\cite{mari2024online}.

\emph{Beyond additive delays.} Moving beyond additive delay costs, Dooly, Goldman, and Scott studied TCP acknowledgment in which each batch incurs the maximum waiting time of its packets, and established the optimal deterministic competitive ratio of $2$~\cite{dooly2001tcp}.
Bhore, Paw{\l}owski, and Umboh generalized the delay cost assigned to each batch to any monotone function~\cite{bhore2026online}.
When these costs are summed across batches, they showed that greedy can be $\Omega(n)$-competitive, while the optimal deterministic ratio over this general class is $\Theta(\log n)$.
Another formulation charges only the maximum waiting time over all packets.
For this global maximum delay objective, Albers and Bals obtained the tight deterministic ratio $\pi^2/6$ and a randomized lower bound of $3/(3-2/e)$~\cite{albers2005dynamic}.

\paragraph{Organization.}
Section~\ref{sec:preliminaries} defines the problem and the competitive ratios.
Section~\ref{sec:offline} presents the offline dynamic program and guiding example, then proves its correctness and analyzes its complexity.
Section~\ref{sec:deterministic} presents the deterministic algorithm and proves the deterministic bound.
Section~\ref{sec:randomized} presents the randomized algorithm and proves the randomized bound.
Section~\ref{sec:extensions} extends the results to models with submodular service costs, and Section~\ref{sec:conclusion} concludes.

\section{Preliminaries}
\label{sec:preliminaries}

We study Multi-Level Aggregation with Maximum Delay (MLA-MD) on a finite rooted tree $T=(V,E)$ with root $\rho\in V$ and nonnegative edge weights $w_e$.
The tree and its edge weights are fixed and known before the input starts.
For a connected subtree $S$ containing $\rho$, write $w(S)=\sum_{e\in E(S)}w_e$.
Let $P(\rho,v)$ denote the unique root path to $v$.
An input $I$ is a finite collection of labelled request occurrences, each with an arrival time $a(r)\ge0$ and a location $v(r)\in V$.
Requests with the same arrival time are revealed together, and each request remains pending until a service clears it.
At time $t$, a service chooses a connected rooted subtree $S$ and clears all pending requests located in $S$.
If the cleared set $B$ is nonempty, the service costs $\cost(B,t,S)=w(S)+\max_{r\in B}(t-a(r))$.
A schedule is feasible if every request is eventually cleared, and its total cost is $\sum_s\cost(B_s,t_s,S_s)$ over its services.
An infeasible schedule has cost $+\infty$, and $\OPT(I)$ denotes the minimum cost of a feasible offline schedule.
If all edge weights are zero, serving every request immediately upon arrival is optimal and has zero total cost.
Throughout the paper, we assume that the input tree has at least one positive-weight edge.

At time $t$, an online decision may use only the tree, the arrival and service history revealed by $t$, elapsed time, and internal random bits.
Future arrivals and an end-of-input signal are unavailable.
A deterministic algorithm is $\alpha$-competitive if $\ALG(I)\le\alpha\OPT(I)$ for every input $I$.
For a randomized algorithm, the random seed is sampled independently of the input.
Against an oblivious adversary, it is $\alpha$-competitive if $\E[\ALG(I)]\le\alpha\OPT(I)$ for every input fixed independently of its random outcome.
The expectation is over internal randomness only, and both definitions are strictly multiplicative, with no additive constant.

\section{Offline Algorithm for MLA-MD}
\label{sec:offline}

We compute the offline optimum and an optimal service schedule for the offline version of MLA-MD.
Subsection~\ref{subsec:offline-uncrossing} describes the algorithm and a guiding example.
Subsection~\ref{subsec:offline-dp} proves correctness and analyzes the implementation complexity.

\subsection{Algorithm Description}
\label{subsec:offline-uncrossing}

Let $a_1<a_2<\cdots<a_n$ be the distinct arrival times, and let $R_i$ be the set of requests arriving at $a_i$.
For $i\le j$, set $R[i,j]=\bigcup_{h=i}^{j}R_h$.
For a finite request set $A$, let $U(A)=\{\rho\}\cup\bigcup_{r\in A}P(\rho,v(r))$ denote its rooted footprint.
Define the interval spatial cost by $C(i,j)=w(U(R[i,j]))$.
For $j=0,\ldots,n$, let $F(j)$ denote the prefix state associated with the first $j$ arrival epochs.
A service schedule is represented by an ordered list of pairs $(t,S)$, each specifying a service of the rooted subtree $S$ at time $t$.

Our algorithm is given in Algorithm~\ref{alg:offline-dp}.
Our idea is to compute the offline optimum by dynamic programming, maintaining a table of prefix states and recording the choices needed to recover an optimal service schedule.
The algorithm takes the tree and the input $I$.
We first sort the requests by arrival time, identify the epochs $a_i$ and request sets $R_i$, and initialize $F(0)=0$.
We fill the table in increasing order of $j$.
For each $j$, we scan $p=j,j-1,\ldots,1$, maintaining the rooted union $U(R[p,j])$ to compute the spatial costs $C(p,j)$, and update the state by
\begin{equation}
 \label{eq:offline-prefix-dp}
F(j)=\min_{1\le p\le j} \left\{F(p-1)+C(p,j)+a_j-a_p\right\}.
\end{equation}
In this update, we consider extending the schedule for state $F(p-1)$ with a service $(a_j,U(R[p,j]))$, whose added cost is $C(p,j)+a_j-a_p$.
We store a minimizing index $p_j^*$ for each state and backtrack from $j=n$ to recover the service times and subtrees.
Finally, we return $F(n)$ together with the recovered schedule.

\begin{algorithm}[H]
\caption{Offline Dynamic Program}
\label{alg:offline-dp}
\begin{algorithmic}[1]
\Require A finite rooted tree $T=(V,E)$ and an input $I$
\Ensure $\OPT(I)$ and an optimal feasible service schedule $\mathcal S$
\State Identify the distinct arrival times $a_1<\cdots<a_n$ and their request sets $R_i$
\If{$n=0$}
  \State \Return $(0,\varnothing)$
\EndIf
\State $F(0)\gets0$
\For{$j=1,\ldots,n$}
  \State $U\gets\{\rho\}$
  \For{$p=j,j-1,\ldots,1$}
    \State Add the root paths of all requests in $R_p$ to $U$
    \State $C(p,j)\gets w(U)$
  \EndFor
  \State Update $F(j)$ by Equation~\eqref{eq:offline-prefix-dp} and store a minimizing index $p_j^*$
\EndFor
\State $\mathcal S\gets$ empty list; $j\gets n$
\While{$j\ge1$}
  \State $p\gets p_j^*$
  \State Prepend $(a_j,U(R[p,j]))$ to $\mathcal S$; $j\gets p-1$
\EndWhile
\State \Return $(F(n),\mathcal S)$
\end{algorithmic}
\end{algorithm}

\paragraph{Example.}
Consider the finite rooted path with vertices at coordinates $0,1,2,3,4$, unit edge weights, and root at $0$.
The four requests and the resulting dynamic-program values are summarized in Table~\ref{tab:offline-path-example}.
On this path, a joint service reaches the farthest request, so $C(i,j)=\max_{i\le k\le j}x_k$.
In particular, recurrence~\eqref{eq:offline-prefix-dp} gives $F(4)=8$.
The algorithm returns one service at time $a_4=4$ covering all four requests, with spatial cost $4$ and maximum waiting time $4$.

\begin{table}[t]
\centering
\caption{Input data and dynamic-program values for the path example.}
\label{tab:offline-path-example}
\setlength{\tabcolsep}{1.15em}
\renewcommand{\arraystretch}{1.12}
\begin{tabular}{@{}lcccc@{}}
\toprule
& \multicolumn{4}{c}{Request index $j$} \\
\cmidrule(l){2-5}
Quantity & $1$ & $2$ & $3$ & $4$ \\
\midrule
Arrival time $a_j$    & $0$ & $1$ & $\tfrac{5}{2}$  & $4$ \\
Location $x_j$        & $3$ & $1$ & $4$              & $2$ \\
Spatial cost $C(1,j)$ & $3$ & $3$ & $4$              & $4$ \\
Prefix optimum $F(j)$ & $3$ & $4$ & $\tfrac{13}{2}$ & $8$ \\
\bottomrule
\end{tabular}
\end{table}

\subsection{Analysis}
\label{subsec:offline-dp}

In this subsection, we prove that Algorithm~\ref{alg:offline-dp} returns the offline optimum and an optimal feasible service schedule in polynomial time.
We first relate feasible schedules to an auxiliary optimization over request partitions. We then merge batches with intersecting arrival-time intervals to obtain an optimal schedule whose batches consist of consecutive complete arrival epochs.
We use this structure to prove the recurrence and schedule recovery, and then analyze the running time.

A \emph{service batch} is the set of requests cleared by one service.
We treat requests as labelled occurrences, so distinct requests remain distinct even when they have the same location.
For a finite set $A$ of such occurrences, define its minimum spatial service cost by $\kappa(A)=w(U(A))$.
The set $U(A)$ is the unique inclusion-minimal rooted subtree covering $A$, so $\kappa(A)$ is the minimum spatial cost of a service that serves all of $A$.
For an epoch interval, this gives $C(i,j)=\kappa(R[i,j])$.
For a request set $B$, let $\alpha(B)$ and $\beta(B)$ be its earliest and latest arrival times, and let $H(B)=[\alpha(B),\beta(B)]$ denote the closed interval from $\alpha(B)$ to $\beta(B)$.
The ideal cost of $B$ is $\phi(B)=\kappa(B)+\beta(B)-\alpha(B)$.
For a partition $\Pi$ of all request occurrences into nonempty sets, set $\Phi(\Pi)=\sum_{B\in\Pi}\phi(B)$ and $\Phi^\star(I)=\min_{\Pi}\Phi(\Pi)$.
The next proposition states that the offline optimum equals the minimum total ideal cost over all request partitions.

\begin{proposition}
 \label{prop:consecutive-normal-form}
For every input $I$, the offline optimum satisfies $\OPT(I)=\Phi^\star(I)$.
\end{proposition}

\begin{proof}
We first compare a feasible schedule with its service-batch partition.
For a service at time $t$ that clears $B$, feasibility gives $t\ge\beta(B)$ and the service subtree has weight at least $\kappa(B)$.
Thus its cost is at least $\phi(B)$.
Consequently $\OPT(I)\ge\Phi^\star(I)$.

Choose a partition attaining $\Phi^\star(I)$ and repeatedly merge two batches $B$ and $B'$ whose intervals $H(B)$ and $H(B')$ intersect.
Since $U(B\cup B')=U(B)\cup U(B')$ and edge weights are nonnegative, $\kappa(B\cup B')\le\kappa(B)+\kappa(B')$.
For intersecting intervals, this spatial bound and subadditivity of interval length give $\phi(B\cup B')\le\phi(B)+\phi(B')$.
The process ends with pairwise disjoint intervals.
Each merge does not increase the partition cost.  Since the starting partition
already attains the minimum $\Phi^\star(I)$, the merged partition still has
cost exactly $\Phi^\star(I)$.

Such a partition cannot split an arrival epoch, since both batches would then contain its arrival time in their intervals.
Nor can one batch contain requests from epochs $p$ and $q$ while an intermediate epoch $h$ belongs to another batch.
The interval of that other batch contains $a_h$, which already lies in $[a_p,a_q]$.
Thus each batch of the merged partition consists of requests from consecutive complete arrival epochs.
Serve the batches of the merged partition from left to right.
For a batch $R[p,q]$, serve it at $a_q$ using $U(R[p,q])$.
Earlier requests have been cleared, and later epochs have not arrived, so this service clears exactly the current batch.
Its spatial cost is $C(p,q)$ and its oldest request waits $a_q-a_p$, so its cost is $C(p,q)+a_q-a_p$.
The realized schedule has cost $\Phi^\star(I)$, so $\OPT(I)\le\Phi^\star(I)$.
Equality follows.
\end{proof}

We now prove that Algorithm~\ref{alg:offline-dp} solves the offline problem.
For $i\le j$, let $D(i,j)$ be the minimum cost of a feasible service schedule for the subinstance $R[i,j]$.
For empty intervals, set $R[i,i-1]=\varnothing$ and $C(i,i-1)=D(i,i-1)=0$.
\begin{theorem}
 \label{thm:offline-dp}
For every input $I$, Algorithm~\ref{alg:offline-dp} returns $\OPT(I)$ and an optimal feasible service schedule.
For every $1\le i\le j\le n$, the subinstance optimum satisfies
\begin{equation}
 \label{eq:offline-dp}
D(i,j)=\min_{i\le p\le j} \left\{D(i,p-1)+C(p,j)+a_j-a_p\right\}.
\end{equation}
\end{theorem}

\begin{proof}
The empty input has optimum zero and is handled directly by the algorithm.
For a nonempty subinstance $R[i,j]$, the construction in the proof of Proposition~\ref{prop:consecutive-normal-form} yields an optimal schedule whose service batches consist of requests from consecutive complete arrival epochs.
If its final batch is $R[p,j]$, the preceding services cost at least $D(i,p-1)$ and the final service costs $C(p,j)+a_j-a_p$.
Taking the minimum over $p$ gives the lower bound in~\eqref{eq:offline-dp}.
Conversely, for each $p$, the same construction gives an optimal schedule for $R[i,p-1]$ completing by $a_{p-1}$ when this subinstance is nonempty.
Append a service at $a_j$ using $U(R[p,j])$.
All earlier requests have been cleared and all later requests have yet to arrive, so the appended service clears exactly $R[p,j]$ and costs $C(p,j)+a_j-a_p$.
Every candidate is therefore attainable, proving the recurrence.

Comparing~\eqref{eq:offline-prefix-dp} with~\eqref{eq:offline-dp}, induction on $j$ gives $F(j)=D(1,j)$ and hence $F(n)=\OPT(I)$.
Each step of backtracking records a service at $a_j$ before decreasing $j$ to $p_j^*-1$.
The recovered services are ordered chronologically by prepending them to $\mathcal S$.
At a recovered service $(a_j,U(R[p,j]))$, earlier recovered services have cleared $R[1,p-1]$, so the pending set is exactly $R[p,j]$.
Thus every request is served, and each service incurs exactly the corresponding transition cost.
These costs sum to $F(n)$, proving that the returned schedule is feasible and optimal.\qedhere
\end{proof}

We next analyze the complexity of Algorithm~\ref{alg:offline-dp}.
\begin{theorem}
 \label{thm:offline-polytime}
Algorithm~\ref{alg:offline-dp} runs in time polynomial in the input size.
\end{theorem}

\begin{proof}
We use the unit-cost arithmetic model and let $N$ be the number of request occurrences.
The algorithm returns immediately if there are no arrival epochs.
Otherwise, let $n$ denote the number of arrival epochs.
Sorting the requests and identifying the distinct arrival epochs takes $O(N\log(N+1))$ time.

For a fixed endpoint $j$, scan $p=j,j-1,\ldots,1$, maintaining the marked rooted union of requests in $R[p,j]$ and its weight.
When an epoch is prepended, follow each request path toward the root until it reaches the marked union, adding only newly marked edges.
Each edge is marked at most once during the scan, and each request requires one final membership check.
Thus all costs $C(p,j)$ for this endpoint are computed in $O(|E|+N)$ time.

Once these costs are available, computing $F(j)$ from~\eqref{eq:offline-prefix-dp} tests $j$ predecessors and takes $O(j)$ time.
Over all endpoints, the scans and state updates therefore take $O(n^2+nN+n|E|)$ time and use $O(n^2)$ interval-cost evaluations.
Storing one minimizing predecessor per state allows the service times and implicit subtrees to be recovered in $O(n)$ time.
Each service subtree is represented by its request interval and the corresponding rooted footprint.
The states and backpointers use $O(n)$ space, and the total working storage is $O(|E|+N+n)$, including edge marks and input lists.
Including preprocessing, the running time is $O(N\log(N+1)+n^2+nN+n|E|)$.
Since $n\le N$, this is $O(N^2+N|E|)$.
The tree and requests are explicitly represented, so both $N$ and $|E|$ are bounded by the input size.
Hence the running time is polynomial in the input size.
\end{proof}

\section{Online Deterministic Algorithm for MLA-MD}
\label{sec:deterministic}

We design a deterministic $2$-competitive online algorithm for MLA-MD.
Subsection~\ref{subsec:opt-timer} presents the parameterized algorithm $\OPTTimer(\theta)$.
Our deterministic algorithm sets $\theta=1$.
Subsection~\ref{subsec:deterministic-analysis} analyzes this algorithm and proves that $\OPTTimer(1)$ is $2$-competitive.
Together with the matching lower bound in Appendix~\ref{app:known-lower-bounds}, this establishes that the deterministic competitive ratio of $2$ is optimal.

\subsection{Algorithm Description}
\label{subsec:opt-timer}

The parameterized algorithm is given in Algorithm~\ref{alg:opt-timer}.
For the deterministic algorithm, we set $\theta=1$.
A \emph{busy period} begins when an epoch arrives while no request is pending and ends with the algorithm's next service.
The consecutive complete arrival epochs accumulated during one OPT-Timer busy period form a \emph{busy block}.
We call this algorithm \emph{OPT-Timer}, as it uses the offline optimum of the currently observed busy block to determine its service deadline.
For a current block, the offline value measures the cost of serving its requests optimally, while the cost of one joint service measures the spatial part of serving them together.
We use their difference as a temporal offset in the service deadline.
We add a fraction $\theta$ of the joint-service cost to this offset.

To implement this idea, we first serve requests whose root paths have zero cost immediately upon arrival and remove them.
If this empties an epoch, we delete the epoch and reindex.
Because edge weights are nonnegative, this preprocessing is cost-free and cannot increase the cost of any offline batch.
We fix $\theta\in[0,1]$ for the entire input.
For a busy period beginning at epoch $s$ and currently ending at epoch $j$, let $B=[s,j]$.
Using the revealed prefix $R[s,j]$, we compute the joint-service cost $C(s,j)$ from its rooted footprint and update the offline value $D(s,j)$ by Equation~\eqref{eq:offline-dp}.
We record the resulting offset and define the deadline by
\begin{equation}
 \label{eq:envelope-slack}
C(B)=C(s,j),\qquad D(B)=D(s,j),\qquad \Env(B)=D(B)-C(B).
\end{equation}
\begin{equation}
 \label{eq:theta-deadline}
\tau_\theta(B)=a_s+\Env(B)+\theta C(B) =a_s+D(B)-(1-\theta)C(B).
\end{equation}
We wait for the next arrival epoch or the deadline, whichever occurs first.
If the next epoch arrives no later than the deadline, we append the complete epoch to the busy block, extend the state table, and recompute the deadline.
Otherwise, we serve $U(R[s,j])$ at the deadline and end the busy period.
The rule uses only revealed epochs and does not require an end-of-input signal.
When an arrival coincides with a deadline, we process the arriving epoch first, as specified in Algorithm~\ref{alg:opt-timer}.
We append the whole epoch and recompute the deadline.
If the new deadline equals the current time, we serve the pending requests immediately afterward.

\begin{algorithm}[!t]
\caption{$\OPTTimer(\theta)$}
\label{alg:opt-timer}
\begin{algorithmic}[1]
\Require A rooted tree $T=(V,E)$, a fixed parameter $\theta\in[0,1]$, and an input $I$ revealed online
\Ensure A feasible service schedule for $I$
\While{the online process is active}
  \State Wait until a complete epoch $R_s$ arrives at time $a_s$ while the system is idle
  \State $j\gets s$ \Comment{start the busy block $B=[s,j]$}
  \While{the current busy period is active}
    \State Compute $C(s,j)$ and $D(s,j)$ from the revealed epochs $R[s,j]$
    \State $\tau\gets a_s+D(s,j)-(1-\theta)C(s,j)$
    \State Wait for the next arrival epoch or the timer at $\tau$, whichever occurs first
    \If{an epoch $R_{j+1}$ arrives at time $a_{j+1}\le\tau$}
      \State $j\gets j+1$ \Comment{arrival-first rule, including equality}
    \Else
      \State Serve $U(R[s,j])$ at time $\tau$ and end the busy period
      \State \textbf{break}
    \EndIf
  \EndWhile
\EndWhile
\end{algorithmic}
\end{algorithm}

Note that $\OPTTimer(\theta)$ can be computed in polynomial time using only revealed requests for any $\theta\in[0,1]$.
Within a busy block beginning at epoch $s$, whenever epoch $j$ is accepted, a backward scan computes $C(p,j)$ for every $s\le p\le j$ and updates $D(s,j)$ using recurrence~\eqref{eq:offline-dp}.
Applying the bounds from Theorem~\ref{thm:offline-polytime} over all busy blocks gives total running time $O(n^2+nN+n|E|)$.

\begin{figure}[H]
\centering
 \begin{tikzpicture}[x=1.18cm,y=0.82cm,>=Latex,font=\small]
\draw[->,thick] (-0.2,0) -- (8.8,0) node[right] {time}; \draw[->,thick] (0,-0.15) -- (0,5.0) node[above] {location}; \foreach \y in {1,2,3,4} \draw[gray!35] (0,\y) -- (8.45,\y); \foreach \x/\lab in {1/1,2.5/{5/2},4/4,8/8} \draw (\x,0.08) -- (\x,-0.08) node[below=2pt] {$\lab$}; \foreach \y in {1,2,3,4} \draw (0.08,\y) -- (-0.08,\y) node[left=2pt] {$\y$};

\draw[densely dashed,gray] (0,3) -- (8,3); \draw[densely dashed,gray] (1,1) -- (8,1); \draw[densely dashed,gray] (2.5,4) -- (8,4); \draw[densely dashed,gray] (4,2) -- (8,2);

\fill[red!75!black] (0,3) circle (2.2pt) node[above right=-1pt] {$r_1$}; \fill[red!75!black] (1,1) circle (2.2pt) node[above right=-1pt] {$r_2$}; \fill[red!75!black] (2.5,4) circle (2.2pt) node[above right=-1pt] {$r_3$}; \fill[red!75!black] (4,2) circle (2.2pt) node[above right=-1pt] {$r_4$};

\draw[blue!65!black,line width=2.2pt] (8,0) -- (8,4); \draw[decorate,decoration={brace,mirror,amplitude=5pt},red!70!black] (0,-0.72) -- (8,-0.72) node[midway,below=7pt] {maximum waiting time $8$};

 \end{tikzpicture}
 \caption{Execution of OPT-Timer on the rooted path for the guiding example in Subsection~\ref{subsec:offline-uncrossing} with $\theta=1$.
Requests are points in the time--location plane.
The terminal vertical segment is the rooted service, and the horizontal dashed segments represent the requests' waiting times.}
 \label{fig:line-space-time}
\end{figure}

\paragraph{Example.}
For the guiding example in Subsection~\ref{subsec:offline-uncrossing}, $D(1,j)=F(j)$.
For $\theta=1$, the successive deadlines from \eqref{eq:theta-deadline} are $3,4,13/2,$ and $8$.
Each later request arrives before the current deadline, so all four requests enter one busy block.
Figure~\ref{fig:line-space-time} shows the resulting service at time $8$ with spatial cost $4$.
The first request determines the maximum waiting time, while the third determines the rooted footprint.
Thus the realized service cost is $4+8=12$, whereas the offline optimum for the same instance costs $8$.

\subsection{Analysis}
\label{subsec:deterministic-analysis}

We analyze the deterministic algorithm in three steps.
First, we establish the spatial properties of the rooted service cost.
Next, we use them to prove the deadline invariant and the exact cost of each completed busy block.
Finally, we derive the interval structure of the timer-induced partitions and use it to prove the competitive ratio.
The following proposition records the properties of $\kappa$.
\begin{proposition}
 \label{prop:spatial-cost-properties}
The function $\kappa$ on labelled request sets has the following properties.
\begin{enumerate}
\renewcommand{\theenumi}{\roman{enumi}}
\renewcommand{\labelenumi}{(\theenumi)}
 \item \label{item:kappa-basic} It is normalized, nonnegative, and monotone.
 \item \label{item:kappa-subadditive} It is subadditive and submodular.
 \item \label{item:kappa-submodular} For all labelled request sets $A,B$,
 \begin{equation}
 \label{eq:submodularity}
\kappa(A)+\kappa(B) \ge \kappa(A\cup B)+\kappa(A\cap B).
 \end{equation}
 \item \label{item:kappa-diminishing-returns} For all labelled request sets $A\subseteq B$ and $Y$,
$$
\kappa(A\cup Y)-\kappa(A) \ge \kappa(B\cup Y)-\kappa(B).
$$
\end{enumerate}
\end{proposition}

\begin{proof}
For each edge $e$, let $V_e$ be the component away from the root after deleting $e$, and let $\chi_e(A)$ be the indicator that some request in $A$ lies in $V_e$.
Then $\kappa(A)=\sum_{e\in E}w_e\chi_e(A)$.
Each indicator $\chi_e$ is normalized, nonnegative, monotone, and submodular on request-occurrence sets.
Nonnegative weighted sums preserve all four properties, proving part~(\ref{item:kappa-basic}) and the submodularity assertion of part~(\ref{item:kappa-subadditive}).
Part~(\ref{item:kappa-submodular}) is the submodularity inequality.
Combining it with $\kappa(A\cap B)\ge0$ gives
$$
\kappa(A\cup B)\le\kappa(A)+\kappa(B)-\kappa(A\cap B)\le\kappa(A)+\kappa(B),
$$
which proves the subadditivity assertion of part~(\ref{item:kappa-subadditive}).

For part~(\ref{item:kappa-diminishing-returns}), fix finite labelled request sets $A\subseteq B$ and $Y$.
Apply part~(\ref{item:kappa-submodular}) to $A\cup Y$ and $B$.
Their union is $B\cup Y$, and their intersection contains $A$.
By monotonicity,
$$
\kappa(A\cup Y)+\kappa(B)\ge\kappa(B\cup Y)+\kappa((A\cup Y)\cap B)\ge\kappa(B\cup Y)+\kappa(A).
$$
Rearranging proves part~(\ref{item:kappa-diminishing-returns}).
\end{proof}

We next compare the interval optimum with the cost of serving all requests in the interval once.
\begin{lemma}
 \label{lem:offline-basic-bounds}
For every input $I$ with epochs $a_1<\cdots<a_n$ and every $1\le i\le j\le n$,
 \begin{equation}
 \label{eq:offline-basic-bounds}
C(i,j)\le D(i,j)\le C(i,j)+a_j-a_i.
 \end{equation}
\end{lemma}

\begin{proof}
For any feasible schedule serving $R[i,j]$, let $B_1,\ldots,B_s$ be its nonempty service batches.
Proposition~\ref{prop:spatial-cost-properties}(\ref{item:kappa-subadditive}) gives
$$
C(i,j)=\kappa\!\left(\bigcup_{h=1}^sB_h\right)\le\sum_{h=1}^s\kappa(B_h).
$$
Each service pays at least its batch's minimum spatial cost and a nonnegative waiting cost, so every feasible schedule costs at least $C(i,j)$.
Minimizing proves the lower bound in \eqref{eq:offline-basic-bounds}.
One service at $a_j$ using $U(R[i,j])$ clears all requests in the subinstance at cost $C(i,j)+a_j-a_i$, proving the upper bound.
\end{proof}

The following lemma gives the exact cost of a completed busy block.
\begin{lemma}
 \label{lem:timer-properties}
For fixed $\theta\in[0,1]$, the total cost of each busy block $B$ completed by $\OPTTimer(\theta)$ is $D(B)+\theta C(B)$.
\end{lemma}

\begin{proof}
We first show that, after each accepted epoch $j$ of a busy period beginning at epoch $s$, the pending set is $R[s,j]$, the deadline is no earlier than $a_j$, and accepting another epoch never moves the deadline earlier.
For two consecutive accepted epochs $j$ and $j+1$, put $X=R[s,j]$, $Y=R_{j+1}$, and $\Delta=\kappa(X\cup Y)-\kappa(X)$.
Removing $Y$ from the final block of an optimal partition for $[s,j+1]$ decreases its spatial cost by at least $\Delta$ by Proposition~\ref{prop:spatial-cost-properties}(\ref{item:kappa-diminishing-returns}), with the monotonicity in Proposition~\ref{prop:spatial-cost-properties}(\ref{item:kappa-basic}) covering the case in which that block is just $Y$.
Its temporal cost does not increase.
Thus $D(s,j)\le D(s,j+1)-\Delta$ and $C(s,j+1)=C(s,j)+\Delta$, so $\Env(s,j+1)\ge\Env(s,j)$.

At the start of a busy period, $R_s$ is inserted before any service.
Since $D(s,s)=C(s,s)$, the initial deadline is $a_s+\theta C(s,s)\ge a_s$.
Suppose inductively that epoch $j+1$ arrives no later than the current deadline.
No service has occurred, so appending the whole epoch changes the pending set from $R[s,j]$ to $R[s,j+1]$.
The monotonicity above and $\theta\ge0$ show that the new deadline is no earlier than the old one, hence no earlier than $a_{j+1}$.
If no epoch arrives by the deadline, the algorithm clears the entire pending set at the deadline.
This proves the claim.
Hence no request in the busy period is served before the terminal service, and the first epoch remains pending until that service.
Consequently the maximum waiting time is $\tau_\theta(B)-a_s=D(B)-(1-\theta)C(B)$.
The terminal footprint has spatial cost $C(B)$.
Adding the two terms proves the lemma.
\end{proof}

The timer-induced partitions have the following interval structure.
Set $\lambda_\theta=1-\theta$ and define $h_\theta(i,j)=D(i,j)-\lambda_\theta C(i,j)$, with $h_\theta(i,i-1)=0$.
Call $[i,k]$ \emph{$\theta$-connected} when $a_r-a_i\le h_\theta(i,r-1)$ for every $r\in\{i+1,\ldots,k\}$.
This is exactly the condition for the timer to accept all epochs through $k$ into a block beginning at $i$.
The \emph{$\theta$-greedy partition} takes the longest $\theta$-connected prefix and repeats on the remaining suffix.
At a nonfinal boundary $k$, arrival-first rejection gives
\begin{equation}
 \label{eq:theta-greedy-cut}
a_k-a_i>h_\theta(i,k-1),
\end{equation}
where $i$ is the block's first epoch.
A partition is \emph{weakly $\theta$-greedy} if its blocks are $\theta$-connected and its nonfinal boundaries satisfy the weak version of \eqref{eq:theta-greedy-cut}.
The next lemma is the combinatorial core.
It shows that a connected prefix can be extended and that a valid cut can be propagated through the remaining suffix.
\begin{lemma}
 \label{lem:extension-cut}
Fix $\theta\in[0,1]$ and an epoch interval $[i,j]$.
\begin{enumerate}
\renewcommand{\theenumi}{\roman{enumi}}
\renewcommand{\labelenumi}{(\theenumi)}
 \item \label{item:greedy-extension} For every $k$ with $i\le k\le j$, if $[i,k]$ is $\theta$-connected, then
 \begin{equation}
 \label{eq:extension-inequality}
h_\theta(i,j)\ge a_k-a_i+h_\theta(k,j).
 \end{equation}
 \item \label{item:greedy-cut} For every $k$ with $i<k\le j$, if $[i,k-1]$ is $\theta$-connected and
$a_k-a_i\ge h_\theta(i,k-1)$, then
 \begin{equation}
 \label{eq:cut-inequality}
h_\theta(i,j)\ge h_\theta(i,k-1)+h_\theta(k,j).
 \end{equation}
 \item \label{item:greedy-superadditivity} If $B_1,\ldots,B_m$ form a weakly $\theta$-greedy partition of $[i,j]$, then
 \begin{equation}
 \label{eq:greedy-superadditivity}
h_\theta(i,j)\ge\sum_{\ell=1}^{m}h_\theta(B_\ell).
 \end{equation}
\end{enumerate}
\end{lemma}

\begin{proof}
We first prove parts~(\ref{item:greedy-extension}) and~(\ref{item:greedy-cut}).
Fix $k$ with $i\le k\le j$.
Use the construction in the proof of Proposition~\ref{prop:consecutive-normal-form} to choose an optimal consecutive-block partition of $[i,j]$, and let $[p,q]$ be the block containing $k$.
The suffix after $q$ is unchanged if this block is replaced by $[k,q]$.
Serving $[k,q]$ at time $a_q$ is feasible for the suffix instance $[k,j]$, so
\begin{equation}
\label{eq:cut-block-cost}
D(i,j)-D(k,j)\ge D(i,p-1)+C(p,q)+a_k-a_p-C(k,q).
\end{equation}
The labelled request sets $R[i,q]$ and $R[k,j]$ have union $R[i,j]$ and intersection $R[k,q]$.
Thus Proposition~\ref{prop:spatial-cost-properties}(\ref{item:kappa-submodular}) gives
\begin{equation}
\label{eq:interval-quadrangle}
C(i,q)+C(k,j)\ge C(i,j)+C(k,q).
\end{equation}
Rearranging \eqref{eq:interval-quadrangle} and then using subadditivity yields
\begin{equation}
\label{eq:cut-coverage-cost}
C(i,j)-C(k,j)\le C(i,q)-C(k,q)
\le C(i,p-1)+C(p,q)-C(k,q).
\end{equation}
Subtracting $\lambda_\theta$ times \eqref{eq:cut-coverage-cost} from \eqref{eq:cut-block-cost}, and using $C(p,q)\ge C(k,q)$, gives
\begin{equation}
\label{eq:unified-extension-cut}
h_\theta(i,j)-h_\theta(k,j)
\ge h_\theta(i,p-1)+a_k-a_p
   +\theta\bigl(C(p,q)-C(k,q)\bigr)
\ge h_\theta(i,p-1)+a_k-a_p.
\end{equation}
If $[i,k]$ is $\theta$-connected, then the definition of connectedness gives
$h_\theta(i,p-1)\ge a_p-a_i$ (with equality when $p=i$).
Equation~\eqref{eq:unified-extension-cut} therefore yields
\eqref{eq:extension-inequality}.
For part~(\ref{item:greedy-cut}), if $p=k$ then \eqref{eq:unified-extension-cut} directly gives the claim.
If $p<k$, the connected prefix $[i,k-1]$ gives
$h_\theta(i,p-1)+a_k-a_p\ge a_k-a_i\ge h_\theta(i,k-1)$,
so the same inequality gives \eqref{eq:cut-inequality}.

For part~(\ref{item:greedy-superadditivity}), apply part~(\ref{item:greedy-cut}) successively at the boundaries of the partition.
At each boundary, connectedness of the preceding block and the weak rejection condition supply the required hypotheses.
Iterating on the remaining suffix proves \eqref{eq:greedy-superadditivity}.
Actual greedy partitions satisfy the stronger strict rejection condition.
\end{proof}
We now show that increasing $\theta$ can only merge adjacent blocks.
\begin{lemma}
 \label{lem:theta-coarsening}
For a fixed input and $0\le\theta_1\le\theta_2\le1$, each block of the $\theta_2$-greedy partition is a union of consecutive blocks of the $\theta_1$-greedy partition.
\end{lemma}

\begin{proof}
We compare the acceptance tests for the two parameters and then iterate over the old blocks.
By Proposition~\ref{prop:spatial-cost-properties}(\ref{item:kappa-basic}), $C\ge0$, so $h_{\theta_2}(i,j)\ge h_{\theta_1}(i,j)$.
The first $\theta_2$ block therefore contains the first $\theta_1$ block, say $[i,u-1]$.
If it also accepts $u$, then for every later epoch $r$ in the next $\theta_1$ block, $a_r-a_u\le h_{\theta_1}(u,r-1)\le h_{\theta_2}(u,r-1)$.
Because $[i,u]$ is $\theta_2$-connected, \eqref{eq:extension-inequality} gives
$$
h_{\theta_2}(i,r-1)\ge a_u-a_i+h_{\theta_2}(u,r-1)\ge a_r-a_i.
$$
Thus the $\theta_2$ block absorbs that whole old block.
Repeating this argument shows that it can stop only at a $\theta_1$ boundary, and the same reasoning applies to the suffix.
\end{proof}

Finally, we prove the main result of this section.
\begin{theorem}
 \label{thm:deterministic-upper}
The $\OPTTimer(1)$ is $2$-competitive on rooted trees.
\end{theorem}

\begin{proof}
We compare $\OPT(I)$ with the busy-block values of $\OPTTimer(1)$.
For a nonempty input, let $B_1,\ldots,B_m$ be the busy blocks produced by the algorithm.
Since $h_1=D$, part~(\ref{item:greedy-superadditivity}) of Lemma~\ref{lem:extension-cut} gives
 \begin{equation}
 \label{eq:d-sum-upper}
\OPT(I)=D(1,n)\ge\sum_{\ell=1}^{m}D(B_\ell).
 \end{equation}
Conversely, concatenate optimal consecutive partitions inside the blocks.
This is feasible for the whole input, so $D(1,n)\le\sum_{\ell=1}^{m}D(B_\ell)$.
Hence
 \begin{equation}
 \label{eq:busy-block-opt-decomposition}
\sum_{\ell=1}^{m}D(B_\ell)=D(1,n)=\OPT(I).
 \end{equation}
Lemmas~\ref{lem:timer-properties} and~\ref{lem:offline-basic-bounds} now imply
$$
\ALG(I)=\sum_{\ell=1}^{m}\bigl(D(B_\ell)+C(B_\ell)\bigr)\le2\sum_{\ell=1}^{m}D(B_\ell)=2\OPT(I).
$$
\end{proof}
\begin{corollary}
 \label{cor:deterministic-tight}
On every rooted tree, the optimal deterministic competitive ratio is $2$.
\end{corollary}

\begin{proof}
Theorem~\ref{thm:deterministic-lower}, based on the matching TCP lower bound of Dooly, Goldman, and Scott~\cite{dooly2001tcp}, applies at any node whose root-path cost is positive.
Together with Theorem~\ref{thm:deterministic-upper}, it proves the claim.
A self-contained proof for this restriction appears in Appendix~\ref{app:known-lower-bounds}.
\end{proof}

\section{Online Randomized Algorithm for MLA-MD}
\label{sec:randomized}

In this section, we present a randomized $e/(e-1)$-competitive online algorithm for MLA-MD against an oblivious adversary.
The algorithm samples one timer parameter before the first request arrives and uses it throughout the input.
The analysis follows the timer-induced partitions as this parameter varies.
Together with the matching lower bound in Appendix~\ref{app:known-lower-bounds}, this establishes the optimal randomized competitive ratio.

Specifically, sample $\Theta\in[0,1]$ before the first arrival according to
\begin{equation}
 \label{eq:theta-distribution}
F(x)=\frac{e^x-1}{e-1}, \qquad f(x)=\frac{e^x}{e-1}, \qquad x\in[0,1].
\end{equation}
Here, $F$ and $f$ denote the cumulative distribution function and probability density function of $\Theta$, respectively.
Our randomized algorithm is precisely $\OPTTimer(\Theta)$.
The same sampled value is used in every busy period and every deadline update.

We analyze the randomized algorithm in three steps.
First, for each fixed parameter, we express the algorithmic cost in terms of the timer-induced partition.
Next, we show that an associated potential has nonnegative jumps as the parameter increases and establish its endpoint bounds.
Finally, we integrate the potential identity to prove the competitive ratio.

Because the adversary is oblivious, we may fix a nonempty input $I$ independently of $\Theta$.
For each $\theta\in[0,1]$, let $P_\theta$ be the consecutive partition of the epochs produced by $\OPTTimer(\theta)$, and let $\ALG_\theta(I)$ be the corresponding cost.
Lemma~\ref{lem:timer-properties} gives
\begin{equation}
 \label{eq:fixed-theta-alg}
\ALG_\theta(I)=\sum_{B\in P_\theta}\bigl(D(B)+\theta C(B)\bigr).
\end{equation}
To compare different parameter values, define the potential
\begin{equation}
 \label{eq:q-potential}
Q(\theta)=\sum_{B\in P_\theta}h_\theta(B)=\sum_{B\in P_\theta}\bigl(D(B)-(1-\theta)C(B)\bigr).
\end{equation}
For each block $B$, $h_\theta(B)$ is its timer length and hence its maximum waiting time.
Thus $Q(\theta)$ is the total temporal cost of the partition $P_\theta$.

We next describe how $P_\theta$ changes with $\theta$.
By Lemma~\ref{lem:theta-coarsening}, increasing $\theta$ can only remove boundaries from the partition.
The test for accepting epoch $k$ into a block beginning at epoch $i$ is
\begin{equation}
 \label{eq:affine-acceptance-test}
a_k-a_i\le D(i,k-1)-C(i,k-1)+\theta C(i,k-1).
\end{equation}
There are finitely many such tests.
Each test with $C(i,k-1)>0$ has at most one equality point, while a test with $C(i,k-1)=0$ is independent of $\theta$.
Consequently, the input induces finitely many breakpoints, and $P_\theta$ is constant on every open interval between consecutive breakpoints.

Fix an interior breakpoint $z$ and choose $\varepsilon>0$ such that $(z-\varepsilon,z+\varepsilon)$ contains no other breakpoint.
Let $P_{z^-}$ and $P_{z^+}$ be the constant partitions immediately to the left and right of $z$, respectively, and define
$$
Q(z^-)=\sum_{B\in P_{z^-}}h_z(B),\qquad Q(z^+)=\sum_{B\in P_{z^+}}h_z(B).
$$
The following lemma shows that merging blocks at a breakpoint cannot decrease the potential.
\begin{lemma}
 \label{lem:q-nonnegative-jumps}
For the fixed input $I$, at every interior breakpoint $z$,
 \begin{equation}
 \label{eq:q-nonnegative-jump}
Q(z^+)\ge Q(z^-).
 \end{equation}
\end{lemma}

\begin{proof}
Fix a block $G\in P_{z^+}$.
By Lemma~\ref{lem:theta-coarsening}, $G$ is a union of consecutive blocks $B_1,\ldots,B_s\in P_{z^-}$.
Taking left limits in the acceptance tests shows that every $B_q$ remains $z$-connected.
Taking left limits in the strict rejection tests shows that each boundary between these blocks satisfies the weak form of \eqref{eq:theta-greedy-cut} at $z$.
Thus $B_1,\ldots,B_s$ form a weakly $z$-greedy partition of $G$.
Part~(\ref{item:greedy-superadditivity}) of Lemma~\ref{lem:extension-cut} gives $h_z(G)\ge\sum_{q=1}^s h_z(B_q)$.
Summing this inequality over all $G\in P_{z^+}$ proves \eqref{eq:q-nonnegative-jump}.
The same argument covers several boundaries that disappear simultaneously.
\end{proof}

The guiding example from Subsection~\ref{subsec:offline-uncrossing} makes this breakpoint structure concrete.
With arrival-first tie breaking, the partition changes at $\theta=1/3$, $3/8$, and $1/2$.
For this instance, the resulting potential is
\[
Q(\theta)=
\begin{cases}
10\theta, & 0\le\theta<1/3,\\
1+9\theta, & 1/3\le\theta<3/8,\\
\frac{5}{2}+7\theta, & 3/8\le\theta<1/2,\\
4+4\theta, & 1/2\le\theta\le1.
\end{cases}
\]
Thus Figure~\ref{fig:q-guiding-example} illustrates the nonnegative jumps established above and the endpoint bounds formalized next.

\begin{figure}[H]
\centering
\begin{tikzpicture}[x=10.4cm,y=0.55cm,>=Latex,font=\small]
 \draw[->,thick] (-0.02,0) -- (1.07,0) node[right] {$\theta$};
 \draw[->,thick] (0,-0.14) -- (0,8.65) node[above] {$Q(\theta)$};
 \foreach \y in {0,2,4,6,8}
   \draw (0.008,\y) -- (-0.008,\y) node[left=2pt] {$\y$};
 \draw (0,0.08) -- (0,-0.08) node[below=2pt] {$0$};
 \draw (0.333333,0.08) -- (0.333333,-0.08) node[below=2pt] {$\frac{1}{3}$};
 \draw (0.375,0.08) -- (0.375,-0.08) node[below=2pt] {$\frac{3}{8}$};
 \draw (0.5,0.08) -- (0.5,-0.08) node[below=2pt] {$\frac{1}{2}$};
 \draw (1,0.08) -- (1,-0.08) node[below=2pt] {$1$};
 \foreach \x in {0.333333,0.375,0.5}
   \draw[gray!45,densely dashed] (\x,0) -- (\x,8.25);

 \draw[blue!65!black,very thick] (0,0) -- (0.333333,3.333333);
 \draw[blue!65!black,very thick] (0.333333,4) -- (0.375,4.375);
 \draw[blue!65!black,very thick] (0.375,5.125) -- (0.5,6);
 \draw[blue!65!black,very thick] (0.5,6) -- (1,8);

 \draw[red!70!black,densely dashed,thick]
   (0.333333,3.333333) -- (0.333333,4);
 \draw[red!70!black,densely dashed,thick]
   (0.375,4.375) -- (0.375,5.125);

 \foreach \x/\y in {0.333333/3.333333,0.375/4.375}
   \draw[blue!65!black,thick,fill=white] (\x,\y) circle (2.4pt);
 \foreach \x/\y in {0/0,0.333333/4,0.375/5.125,0.5/6,1/8}
   \filldraw[blue!65!black] (\x,\y) circle (2.4pt);
\end{tikzpicture}
\caption{The potential $Q(\theta)$ for the guiding example in Subsection~\ref{subsec:offline-uncrossing}.
Light dashed lines mark changes in the timer-induced partition, while the short dashed segments mark the two upward jumps.
Open circles show left limits and filled circles show the values selected by the arrival-first convention.}
\label{fig:q-guiding-example}
\end{figure}

The following lemma gives the endpoint bounds needed for the integration.
\begin{lemma}
 \label{lem:q-endpoint-bounds}
For the fixed input $I$, the one-sided limits $Q(0^+)$ and $Q(1^-)$ exist and satisfy $Q(0^+)\ge0$ and $Q(1^-)\le\OPT(I)$.
\end{lemma}

\begin{proof}
The finiteness of the breakpoint set implies that $P_\theta$ is constant on a right neighborhood of $0$ and on a left neighborhood of $1$.
Since $h_\theta(B)$ is affine in $\theta$ for every fixed block $B$, both one-sided limits exist.
For every block $B$ and every $\theta\in[0,1]$, the lower bound in Lemma~\ref{lem:offline-basic-bounds} gives $h_\theta(B)=D(B)-(1-\theta)C(B)\ge\theta C(B)\ge0$.
Hence $Q(0^+)\ge0$.
Moreover, $P_\theta$ is $\theta$-greedy and hence weakly $\theta$-greedy, so part~(\ref{item:greedy-superadditivity}) of Lemma~\ref{lem:extension-cut} gives $Q(\theta)\le h_\theta(1,n)$.
Since $C(1,n)\ge0$ and $\theta\le1$, we have $h_\theta(1,n)\le D(1,n)=\OPT(I)$, and hence $Q(1^-)\le\OPT(I)$.
\end{proof}

Finally, we prove the main result of this section.
\begin{theorem}
 \label{thm:randomized-upper}
The $\OPTTimer(\Theta)$ is $e/(e-1)$-competitive on rooted trees against an oblivious adversary.
\end{theorem}

\begin{proof}
On every open interval containing no breakpoint, $P_\theta$ is fixed and
 \begin{equation}
 \label{eq:q-derivative}
Q'(\theta)=\sum_{B\in P_\theta}C(B).
 \end{equation}
Combining this with \eqref{eq:fixed-theta-alg} and \eqref{eq:q-potential} gives $Q(\theta)+Q'(\theta)=\ALG_\theta(I)$ and hence
 \begin{equation}
 \label{eq:weighted-potential-derivative}
\frac{d}{d\theta}\bigl(e^\theta Q(\theta)\bigr)=e^\theta\ALG_\theta(I).
 \end{equation}

There are finitely many breakpoints, so $\ALG_\theta(I)$ is bounded and piecewise affine.
Integrate \eqref{eq:weighted-potential-derivative} on every breakpoint-free interval and sum.
If $\mathcal Z$ is the set of interior breakpoints, the endpoint terms give
 \begin{equation}
 \label{eq:piecewise-integral-identity}
\int_0^1e^\theta\ALG_\theta(I)\,d\theta=eQ(1^-)-Q(0^+)-\sum_{z\in\mathcal Z}e^z\bigl(Q(z^+)-Q(z^-)\bigr).
 \end{equation}
The last sum subtracts the upward jumps omitted by integration over the breakpoint-free intervals.
Lemmas~\ref{lem:q-nonnegative-jumps} and~\ref{lem:q-endpoint-bounds} therefore imply
 \begin{equation}
 \label{eq:weighted-alg-upper}
\int_0^1e^\theta\ALG_\theta(I)\,d\theta\le e\OPT(I).
 \end{equation}
Finally, averaging over $\Theta$ using the density in \eqref{eq:theta-distribution} gives
$$
\E[\ALG(I)]=\frac1{e-1}\int_0^1e^\theta\ALG_\theta(I)\,d\theta\le\frac{e}{e-1}\OPT(I).
$$
\end{proof}

\begin{corollary}
 \label{cor:randomized-tight}
Against an oblivious adversary, the optimal randomized competitive ratio on every rooted tree is $\frac{e}{e-1}$.
\end{corollary}

\begin{proof}
Theorem~\ref{thm:randomized-upper} gives the upper bound, and Theorem~\ref{thm:randomized-lower} gives the matching lower bound on every rooted tree.
\end{proof}

\section{Extension to Submodular Service Costs}
\label{sec:extensions}

In this section, we apply the framework of offline interval partitions and $\OPTTimer$ developed in Sections~\ref{sec:offline}--\ref{sec:randomized} to service systems with submodular service costs.
We show that submodularity provides the cost inequalities required by the offline recurrence and the timer-induced partition analysis.
This extends the offline dynamic program, the deterministic $2$-competitive algorithm, and the randomized $e/(e-1)$-competitive algorithm to this broader model.

\begin{definition}
\label{def:submodular-service-cost}
Let $\mathcal U$ be a universe of labelled request occurrences.
A \emph{submodular service-cost model} assigns to every finite set $A\subseteq\mathcal U$ a minimum joint-service cost $\kappa(A)\in\mathbb R_{\ge0}$ with the following properties.
\begin{enumerate}
 \item $\kappa(\varnothing)=0$;
 \item $A\subseteq B$ implies $\kappa(A)\le\kappa(B)$;
 \item $\kappa$ is submodular, i.e., $\kappa(A)+\kappa(B)\ge\kappa(A\cup B)+\kappa(A\cap B)$;
\end{enumerate}

For every finite nonempty $A$ and every $t\ge\max_{r\in A}a(r)$, if the pending set is exactly $A$, a service object attaining $\kappa(A)$ can be executed at time $t$.
It clears $A$, incurs spatial cost $\kappa(A)$ and delay $\max_{r\in A}(t-a(r))$, and has no effect on the availability or cost of later services.
Every service that clears a batch $A$ has spatial cost at least $\kappa(A)$.
\end{definition}

We use three standard consequences of the assumptions on $\kappa$.
The function is subadditive, and it satisfies diminishing returns: for finite request sets $A\subseteq B$ and $Y$,
$$
\kappa(A\cup B)\le\kappa(A)+\kappa(B),
\qquad
\kappa(A\cup Y)-\kappa(A)
\ge
\kappa(B\cup Y)-\kappa(B).
$$
Applied to epoch intervals, submodularity also gives, for $i\le k\le q\le j$,
$$
C(i,q)+C(k,j)\ge C(i,j)+C(k,q).
$$

Using the same uncrossing argument as in Section~\ref{sec:offline}, together with the standard submodular consequences above, we obtain the following offline structure.
\begin{lemma}
\label{lem:abstract-offline-transfer}
For every finite input in a submodular service-cost model, there is an optimal schedule whose batches consist of consecutive complete arrival epochs and in which each batch $R[p,q]$ is served at time $a_q$ by a minimum-cost service object.
Moreover, for every $1\le i\le j\le n$,
$$
D(i,j)
=
\min_{i\le p\le j}
\left\{
D(i,p-1)+C(p,j)+a_j-a_p
\right\},
$$
and
$$
C(i,j)\le D(i,j)\le C(i,j)+a_j-a_i.
$$
\end{lemma}

We next transfer the timer and partition structure.
Set $h_\theta(i,j)=D(i,j)-(1-\theta)C(i,j)$ and $h_\theta(i,i-1)=0$, and use the definitions of $\theta$-connected, $\theta$-greedy, and weakly $\theta$-greedy partitions from Section~\ref{sec:deterministic}.
The arrival-first rule makes the partition produced by the abstract $\OPTTimer(\theta)$ precisely the $\theta$-greedy partition: an epoch $r$ is accepted into a block beginning at $i$ exactly when $a_r-a_i\le h_\theta(i,r-1)$.
For a fixed input, denote this partition by $P_\theta$ and define
$$
Q(\theta)=\sum_{B\in P_\theta}h_\theta(B).
$$

Using the same timer and partition arguments as in Sections~\ref{sec:deterministic} and~\ref{sec:randomized}, we obtain the following online properties.
\begin{lemma}
\label{lem:abstract-online-transfer}
For every fixed input in a submodular service-cost model, the following properties hold.
\begin{enumerate}
\renewcommand{\theenumi}{\roman{enumi}}
\renewcommand{\labelenumi}{(\theenumi)}
 \item \label{item:abstract-block-cost} The deadline of an active busy block never moves earlier when a new epoch is accepted, and every completed block $B$ costs $D(B)+\theta C(B)$.
 \item \label{item:abstract-extension-cut} The extension, cut, and weak-partition inequalities~\eqref{eq:extension-inequality}, \eqref{eq:cut-inequality}, and~\eqref{eq:greedy-superadditivity} hold.
 \item \label{item:abstract-coarsening} If $0\le\theta_1\le\theta_2\le1$, every block of $P_{\theta_2}$ is a union of consecutive blocks of $P_{\theta_1}$.
 \item \label{item:abstract-potential} The partition $P_\theta$ has finitely many breakpoints.
 At every interior breakpoint $z$, $Q(z^+)\ge Q(z^-)$, while the endpoint limits exist and satisfy
 $$
 Q(0^+)\ge0
 \qquad\text{and}\qquad
 Q(1^-)\le\OPT(I).
 $$
\end{enumerate}
\end{lemma}

We can now state the extension of the three upper-bound results.
\begin{theorem}
\label{thm:submodular-framework}
Consider a submodular service-cost model under the maximum-delay objective.
Then:
\begin{enumerate}
\renewcommand{\theenumi}{\roman{enumi}}
\renewcommand{\labelenumi}{(\theenumi)}
 \item The abstract dynamic program based on recurrence~\eqref{eq:offline-dp} computes the optimal value and an optimal service schedule;
 \item $\OPTTimer(1)$ is deterministically $2$-competitive;
 \item Sampling one global parameter $\Theta$ according to~\eqref{eq:theta-distribution} and running $\OPTTimer(\Theta)$ is $e/(e-1)$-competitive against an oblivious adversary.
\end{enumerate}
\end{theorem}

\begin{proof}
The empty input is immediate, so fix a nonempty input.
The first assertion follows from Lemma~\ref{lem:abstract-offline-transfer}: the recurrence considers every possible final consecutive block, and backtracking realizes the selected blocks in chronological order.

For the deterministic assertion, let $B_1,\ldots,B_m$ be the busy blocks of $\OPTTimer(1)$.
Since $h_1=D$, part~(\ref{item:abstract-extension-cut}), specifically~\eqref{eq:greedy-superadditivity}, gives
$$
D(1,n)\ge\sum_{\ell=1}^{m}D(B_\ell).
$$
Conversely, concatenating optimal consecutive-block schedules inside the busy blocks is feasible by Lemma~\ref{lem:abstract-offline-transfer} and the no-side-effect condition, so
$$
D(1,n)\le\sum_{\ell=1}^{m}D(B_\ell).
$$
Thus $\sum_{\ell=1}^{m}D(B_\ell)=\OPT(I)$.
Part~(\ref{item:abstract-block-cost}) and the lower bound $C(B_\ell)\le D(B_\ell)$ now give
$$
\ALG(I)
=
\sum_{\ell=1}^{m}\bigl(D(B_\ell)+C(B_\ell)\bigr)
\le
2\sum_{\ell=1}^{m}D(B_\ell)
=
2\OPT(I).
$$

For the randomized assertion, fix the input independently of $\Theta$.
On every interval containing no breakpoint, $P_\theta$ is fixed and
$$
Q'(\theta)=\sum_{B\in P_\theta}C(B).
$$
Part~(\ref{item:abstract-block-cost}) gives
$$
\ALG_\theta(I)
=
\sum_{B\in P_\theta}\bigl(D(B)+\theta C(B)\bigr)
=
Q(\theta)+Q'(\theta).
$$
Hence
$$
\frac{d}{d\theta}\bigl(e^\theta Q(\theta)\bigr)
=
e^\theta\ALG_\theta(I)
$$
between breakpoints.
If $\mathcal Z$ denotes the finite set of interior breakpoints, integration over the breakpoint-free intervals yields
$$
\int_0^1 e^\theta\ALG_\theta(I)\,d\theta
=
eQ(1^-)-Q(0^+)
-
\sum_{z\in\mathcal Z}
e^z\bigl(Q(z^+)-Q(z^-)\bigr).
$$
Part~(\ref{item:abstract-potential}) bounds the right-hand side by $e\OPT(I)$.
Averaging with the density in~\eqref{eq:theta-distribution} gives
$$
\E[\ALG(I)]
=
\frac{1}{e-1}
\int_0^1 e^\theta\ALG_\theta(I)\,d\theta
\le
\frac{e}{e-1}\OPT(I).
$$
\end{proof}

\section{Conclusion}
\label{sec:conclusion}

We studied MLA-MD on weighted rooted trees.
For the offline problem, we give a polynomial-time algorithm that finds an optimal schedule.
For the online problem, we give a deterministic algorithm with competitive ratio $2$ and a randomized algorithm with ratio $e/(e-1)$ against an oblivious adversary.
Both ratios are best possible, even when all requests arrive at a single node.
Thus branching does not worsen the best possible competitive ratio.
Both online algorithms use the offline optimum for the pending requests to decide when to serve them.
The offline algorithm and both online guarantees also extend to the submodular service-cost model.

Several questions remain open.
One direction is to determine how limiting the number of requests cleared by each service affects the optimal competitive ratios.
Recent work has studied TCP acknowledgment when the delay cost of each batch is given by a general monotone function~\cite{bhore2026online}.
Another direction is to determine whether the techniques developed here extend to MLA under such general delay functions.

\section*{Disclosure of AI Assistance}
OpenAI Codex was used extensively throughout the preparation of this manuscript.
Its assistance included organizing source material, formalizing and adversarially auditing proofs, drafting and revising LaTeX, constructing computational sanity checks, and compiling the paper.
The research questions and directions, the underlying algorithmic ideas, and the source materials were supplied or developed by the authors and their collaborators.
The authors reviewed the resulting arguments and manuscript and take full responsibility for all claims and any remaining errors.

\appendix
\section{Deterministic and Randomized Lower Bounds for MLA-MD}
\label{app:known-lower-bounds}

In this section, we establish the deterministic and randomized lower bounds for MLA-MD.
The deterministic bound is due to Dooly, Goldman, and Scott~\cite{dooly2001tcp}.
A self-contained argument is included here.
We then prove the randomized bound with a finite-support distribution.
Restricting all requests to a node of positive root-path cost transfers both bounds to every rooted tree.
We begin with an exact formula for the offline optimum when all requests arrive at one node.
\begin{lemma}
 \label{lem:common-location-opt}
For an input with $N\ge1$ distinct arrival times, suppose all requests lie at a node $v$ with $c=w(P(\rho,v))>0$.
Let the arrival times be $a_1<\cdots<a_N$.
With $g_i=a_{i+1}-a_i$,
 \begin{equation}
 \label{eq:common-location-opt}
\OPT=c+\sum_{i=1}^{N-1}\min\{g_i,c\}.
 \end{equation}
\end{lemma}

\begin{proof}
We optimize the contribution of each gap independently.
The construction in the proof of Proposition~\ref{prop:consecutive-normal-form} allows us to optimize over consecutive blocks.
The first block pays one setup cost $c$.
At each gap $g_i$, keeping the two adjacent requests together contributes $g_i$ to the temporal span, whereas cutting there contributes one additional setup cost $c$.
These decisions are independent, so the smaller of $g_i$ and $c$ is chosen at every gap, proving \eqref{eq:common-location-opt}.
\end{proof}

\subsection{Deterministic lower bound}

We consider the lower bound for deterministic online algorithms.
Our proof adapts the deterministic TCP lower-bound construction of Dooly, Goldman, and Scott~\cite{dooly2001tcp} to a single positive-cost node.
\begin{theorem}
 \label{thm:deterministic-lower}
On every rooted tree, no deterministic online algorithm has competitive ratio below $2$, even if the tree has only one leaf.
\end{theorem}

\begin{proof}
We construct an input recursively after fixing the deterministic algorithm.
Fix a deterministic algorithm $\mathcal A$, a node $v$ with path cost $c>0$, and $\varepsilon>0$.
We construct an input containing only requests at $v$.
Release the first request at $a_1=0$.
Immediately after request $i$ arrives, and while no earlier request is pending, simulate $\mathcal A$ on the continuation with no further arrivals.
Let $t_i$ be the first time at which it clears that request and put $x_i=t_i-a_i$.
If $i<N$, choose $\delta_i>0$ and set
 \begin{equation}
 \label{eq:lower-arrivals}
a_{i+1}=t_i+\delta_i.
 \end{equation}
Choose the total gap delay $\Delta=\sum_{i=1}^{N-1}\delta_i$ arbitrarily small.
If some $t_i$ does not exist, stopping at that request makes $\mathcal A$ infeasible and its competitive ratio infinite.
Otherwise, exactly one request is pending at each service in the constructed sequence.

Every such service costs at least $c$ in space.
With $S=\sum_{i=1}^{N-1}\min\{x_i,c\}$,
 \begin{equation}
 \label{eq:lower-alg-cost}
\ALG\ge\sum_{i=1}^N(c+x_i)\ge Nc+S.
 \end{equation}
The arrival gaps are $g_i=x_i+\delta_i$.
By Lemma~\ref{lem:common-location-opt} and $\min\{x+y,c\}\le\min\{x,c\}+y$,
 \begin{equation}
 \label{eq:lower-opt-cost}
\OPT\le c+S+\Delta.
 \end{equation}
Since $0\le S\le(N-1)c$, choose $\Delta<(N-1)c$.
The ratio lower bound is then decreasing in $S$, and therefore
 \begin{equation}
 \label{eq:deterministic-lower-ratio}
\frac{\ALG}{\OPT} \ge\frac{Nc+S}{c+S+\Delta} \ge\frac{(2N-1)c}{Nc+\Delta}.
 \end{equation}
Taking $N$ large and then $\Delta/c$ small makes the lower bound in \eqref{eq:deterministic-lower-ratio} exceed $2-\varepsilon$.

Once $\mathcal A$ is fixed, its behavior can be simulated recursively, so all arrival times in \eqref{eq:lower-arrivals} are fixed before the input is presented.
The construction observes no random outcome.
\end{proof}

\subsection{Randomized lower bound}
\label{subsec:randomized-lower}

In this section, we prove the randomized lower bound via Yao's principle by constructing a finite-support distribution over inputs whose requests all lie at a node $v$ with positive path cost $c=w(P(\rho,v))>0$.
Its gap survival probabilities are chosen so that every deterministic service time gives conditional expected charge at least $c$.
As the grid is refined, they converge to the truncated exponential equalizer with density $e^{-g/c}/c$ on $[0,c]$ and termination mass $e^{-1}$, revealing the continuous distribution behind the limiting ratio $e/(e-1)$.

Fix integers $m,K\ge2$, put $\delta=1/m$, and set
\begin{equation}
 \label{eq:lower-grid-distribution}
g_j=\frac{jc}{m}, \qquad p_j=\delta(1-\delta)^{j-1}\quad(j=1,\ldots,m), \qquad s_m=(1-\delta)^m.
\end{equation}
The probabilities satisfy $\sum_{j=1}^mp_j=1-s_m$.
Write
\begin{equation}
 \label{eq:q-m}
q_m=1-s_m=1-\left(1-\frac1m\right)^m.
\end{equation}

The input starts with a request at time $0$.
After each of the first $K-1$ requests, independently choose the next gap $g_j$ with probability $p_j$, or terminate with probability $s_m$.
If request $K$ is generated, terminate there.
The resulting distribution has finite support and at most $K$ requests.
The complete gap sequence may be sampled before execution, so the adversary is oblivious.

For a realized input, index requests in arrival order.
Since all requests are at $v$, every useful service clears a consecutive block.
For a block $r_p,\ldots,r_q$ served at $t$, assign
\begin{equation}
 \label{eq:online-request-accounts}
\chi_i=a_{i+1}-a_i\quad(p\le i<q), \qquad \chi_q=c+t-a_q.
\end{equation}
The charges sum to $c+t-a_p$, the cost of the block when the path is used, and every actual service has spatial cost at least $c$.
Thus
\begin{equation}
 \label{eq:accounts-below-alg}
\sum_i\chi_i\le\ALG(I).
\end{equation}
Arrival-first processing ensures that a request served at its arrival is last in its batch and receives charge $c$.
The following lemma gives the online and offline expectations of the finite game.
\begin{lemma}
 \label{lem:finite-game-costs}
For an input $I$ drawn from the preceding $(m,K)$ distribution, define
 \begin{align}
A_{m,K}&=\sum_{i=0}^{K-1}q_m^i,\label{eq:a-mk}\\ B_{m,K}&=q_m\sum_{i=0}^{K-2}q_m^i+q_m^{K-1}.\label{eq:b-mk}
 \end{align}
Every deterministic online algorithm satisfies $\E_I[\ALG(I)]\ge cA_{m,K}$, while $\E_I[\OPT(I)]=cB_{m,K}$.
\end{lemma}

\begin{proof}
We first establish the per-request bound used to count the expected online charges.
Condition on any history after arrival-first processing at a newly generated request, and suppose another random transition from \eqref{eq:lower-grid-distribution} remains.
If the request is cleared immediately, \eqref{eq:online-request-accounts} assigns charge $c$.
Otherwise let $x>0$ be the time until the deterministic strategy's next service clearing $v$ if no request arrives.
If no such time exists, termination has probability $s_m>0$ and leaves the request unserved, so the expected cost is infinite.
Assume henceforth that $x<\infty$.

Let $G$ be the sampled gap.
If $G\le x$, the next request arrives first (ties are processed first), and the current request receives charge $G$.
If $G>x$ or termination occurs, the service arrives first and the charge is $x+c$.
Thus
 \begin{equation}
 \label{eq:equalizer-lx}
L(x)=\sum_{\ell:g_\ell\le x}p_\ell g_\ell +\left(s_m+\sum_{\ell:g_\ell>x}p_\ell\right)(x+c).
 \end{equation}
Put $S_j=(1-\delta)^j$ and $g_0=0$.
On each interval $(g_j,g_{j+1})$, $L$ has slope $S_j>0$.
On $(g_m,\infty)$ its slope is $S_m=s_m>0$.
At a grid point $g_j$, arrival-first processing and $p_j=\delta S_{j-1}$ give
 \begin{equation}
 \label{eq:equalizer-step}
L(g_j)-L(g_{j-1}) =(g_j-g_{j-1})S_{j-1}-cp_j =\frac{c}{m}S_{j-1}-cp_j=0.
 \end{equation}
Since $L(0)=c$, $L(x)\ge c$ on $[0,c]$, and the positive slope beyond $c$ gives the same bound for all $x\ge0$.

Request $i$ is generated with probability $q_m^{i-1}$.
Each of the first $K-1$ requests, when generated, therefore has conditional expected charge at least $c$.
If request $K$ is generated, it either is served on arrival (charge $c$) or remains pending until its final batch (charge $c+x\ge c$).
Failure to clear it gives infinite cost.
Summing these charges and applying \eqref{eq:accounts-below-alg} proves the claimed online bound.

For the offline value, let $G_i=a_{i+1}-a_i$ when request $i+1$ exists and set $G_i=\infty$ when request $i$ is last, including forced termination at $K$.
Lemma~\ref{lem:common-location-opt} gives pathwise
 \begin{equation}
 \label{eq:terminal-gap-opt}
\OPT(I)=\sum_{i\ \mathrm{generated}}\min\{G_i,c\}.
 \end{equation}
For a nonforced transition, direct summation gives
 \begin{align}
\E[\min\{G,c\}] &=\frac{c}{m}\sum_{j=0}^{m-1}\Pr\!\left[G>\frac{jc}{m}\right]\\ &=\frac{c}{m}\sum_{j=0}^{m-1}(1-\delta)^j =cq_m.
\label{eq:offline-round-expectation}
 \end{align}
The forced terminal request contributes $c$.
Multiplying by the reach probabilities gives the claimed offline identity.
\end{proof}
The limiting distributional ratio yields the randomized lower bound.
\begin{theorem}
 \label{thm:randomized-lower}
On every rooted tree, no randomized online algorithm against an oblivious adversary has competitive ratio below $e/(e-1)$, even if the tree has only one leaf.
\end{theorem}

\begin{proof}
We apply the finite-game ratio to every deterministic strategy induced by a random seed.
Lemma~\ref{lem:finite-game-costs} gives the distributional ratio
 \begin{equation}
 \label{eq:r-mk}
R_{m,K}=\frac{A_{m,K}}{B_{m,K}}
 \end{equation}
for every deterministic strategy.
For fixed $m$, $R_{m,K}\to1/q_m$ as $K\to\infty$, and \eqref{eq:q-m} gives
$$
\frac1{q_m}\longrightarrow\frac1{1-e^{-1}}=\frac{e}{e-1}\qquad(m\to\infty).
$$
Hence, for every $\eta>0$, finite $m,K$ satisfy
 \begin{equation}
 \label{eq:finite-game-ratio-lower}
R_{m,K}>e/(e-1)-\eta.
 \end{equation}

Fix a randomized algorithm and sample the input independently of its full random seed.
Conditioning on the seed leaves a deterministic strategy, so \eqref{eq:finite-game-ratio-lower} applies.
Averaging over the seed and using Tonelli's theorem for nonnegative costs gives
$$
\E_{I,\mathcal A}[\ALG_{\mathcal A}(I)]\ge R_{m,K}\E_I[\OPT(I)].
$$
The distribution has finite support, so some fixed input satisfies $\E_{\mathcal A}[\ALG_{\mathcal A}(I)] \ge R_{m,K}\OPT(I)$.
A uniform ratio $\gamma<e/(e-1)$ would contradict the existence of this input for sufficiently large $m$ and $K$.
This proves the claim under the strict multiplicative convention of Section~\ref{sec:preliminaries}.
\end{proof}

\section{Why Balance Rules Do Not Suffice}
\label{app:balance-separation}

This appendix quantifies limitations of Balance rules.
First, we observe that the classical additive Balance can be polynomially bad when evaluated under the MLA-MD objective.
Besides, we prove that a variant of Balance, \emph{Max-Balance}, which measures urgency by the waiting time of the oldest pending request in each dyadic prefix, has a tight logarithmic competitive ratio.

\subsection{A polynomial lower bound for Balance}

Consider the continuous half-line rooted at $0$.
A service of length $g$ activates $[0,g]$ and costs $g$.
Any finite execution can be represented on a finite weighted rooted path containing the request locations and service endpoints, with the same service costs and waiting costs.
For $g=2^i$, let $W(t,g)$ be the sum of waiting times of pending requests in $[0,g]$.
The classical Balance rule declares $g$ tight when $W(t,g)=g/4$ and serves the largest tight level.
All requests below are at location $1$, have unit rate in this additive potential, and have distinct arrival times.
The following proposition shows how summing waiting times can cause Balance to perform poorly under the maximum-delay objective.
\begin{proposition}
 \label{prop:balance-separation}
The competitive ratio of \emph{Balance} on rooted-line inputs with $N$ requests is $\Omega(\sqrt N)$ under the MLA-MD objective.
\end{proposition}

\begin{proof}
We construct inputs with $k$ bursts of $m$ requests and then set $k=m$.
Fix integers $k,m\ge2$ and set $h=\frac1{8m^2}, \qquad \delta=\frac1{4m}+\frac{m-1}{16m^2}$.
Construct $k$ bursts.
If burst $i$ starts at $s_i$, release its $m$ requests at $s_i,s_i+h,\ldots,s_i+(m-1)h$.
Set $s_1=0$.
After the service triggered by burst $i$, wait $\varepsilon>0$ and start burst $i+1$.
Because Balance is deterministic, this recursion specifies one input.

At the last arrival of a burst, level $1$ has additive potential $h\sum_{q=0}^{m-1}q=\frac{m-1}{16m}<\frac14$, so no request is served before the burst is complete.
Thereafter the potential increases at rate $m$, and at time $s_i+\delta$ it equals $m\delta-h\frac{m(m-1)}2=\frac14$.
No smaller level contains a request, while every larger level has a larger threshold.
Thus level $1$ triggers first and clears the burst.
Under the maximum-delay objective,
 \begin{equation}
 \label{eq:balance-bad-alg}
\ALG=k(1+\delta).
 \end{equation}

For sufficiently small $\varepsilon$, every interarrival gap is below $1$.
Lemma~\ref{lem:common-location-opt} then shows that serving all requests at the final arrival is optimal, with
 \begin{equation}
 \label{eq:balance-bad-opt}
\OPT=1+(k-1)(\delta+\varepsilon)+(m-1)h.
 \end{equation}
Take $k=m=q$ and $\varepsilon=q^{-3}$.
Then $N=q^2$, $\ALG=q(1+\delta)=q+O(1)$, and $\OPT\longrightarrow1+\frac5{16}=\frac{21}{16}$.
Consequently,
$$
\frac{\ALG}{\OPT}\ge\left(\frac{16}{21}-o(1)\right)q=\Omega(\sqrt N).
$$
\end{proof}

The construction in the proof above shows that an adversary can make Balance perform poorly under the maximum-delay objective by releasing requests in closely spaced bursts.
The sum of waiting times of pending requests quickly reaches Balance's service threshold, forcing frequent services and a large total service cost.
The offline optimum can wait and serve all requests together at low cost.

\subsection{A tight logarithmic bound for Max-Balance}

Consider Max-Balance, a variant of Balance that measures urgency by the waiting time of the oldest pending request in each dyadic prefix.
Let $g_i=2^i$ for $i\in\mathbb Z$.
For time $t$ and level $g$, define $M(t,g)=\max\{t-a(r):r\text{ is pending at }t,\ x(r)\le g\}$, with the empty maximum equal to zero.
A level is tight when $M(t,g)=g/4$.
Whenever one or more levels are tight, \emph{Max-Balance} serves the largest tight level.
Arrivals at a service time are processed first.
The next theorem gives the tight asymptotic competitive ratio.
\begin{theorem}
 \label{thm:max-balance-logarithmic}
The competitive ratio of \emph{Max-Balance} on rooted-line inputs with $N$ requests is $\Theta(\log(N+1))$.
\end{theorem}

\begin{proof}
We establish matching upper and lower bounds.
We first prove the upper bound of $O(\log(N+1))$ by charging each active level to the offline optimum.
A service of length $g_i$ at time $t$ is triggered by maximum delay $g_i/4$, so its total cost is $5g_i/4$.
Choose an oldest witness $r$ and set $I_r=[a(r),t],\qquad |I_r|=g_i/4$.
Immediately outside trigger processing, every level satisfies $M(t,g)\le g/4$.
Ages grow continuously, arrivals add age-zero requests, and services can only decrease $M$.
Therefore the witness lies in $(g_{i-1},g_i]$.
Otherwise $M(t,g_{i-1})\ge g_i/4>g_{i-1}/4$.
Witness intervals of distinct length-$g_i$ services are disjoint, since the next witness arrives strictly after the previous service (including under the arrival-first rule at a coincident endpoint).

Fix a level $i$ and charge each witness to the optimal batch containing it.
Let that batch finish at $T$, have length $A$, and have maximum delay $\Delta$.
If $T\in I_r$, then $A\ge x(r)>g_i/2$.
Disjointness implies that one optimal service contains at most one such interval, so the total $g_i$ charged in this case is at most twice its cost.
If instead $T>t$, then $T-\Delta\le a(r)$ and $I_r\subseteq[T-\Delta,T]$.
Intervals charged to one optimal batch are disjoint, hence have total length at most $\Delta$ and total $g_i$ at most $4\Delta$.
Summing both cases gives
 \begin{equation}
 \label{eq:max-balance-one-level}
\sum_{\substack{\text{Max-Balance services}\\\text{at level }i}}g_i \le4\OPT.
 \end{equation}
The algorithmic cost contributed by one level is therefore at most $5\OPT$.

It remains to count relevant levels without using a spatial aspect ratio.
Let $B=\OPT$, let $m_i$ be the number of length-$g_i$ services, and let $g_{i^\star}$ be the largest active level.
Charging one witness shows $g_{i^\star}\le4B$, while distinct services have distinct witnesses, so
 \begin{equation}
 \label{eq:max-balance-witness-count}
\sum_i m_i\le N.
 \end{equation}
For $k=i^\star-i\ge0$,
$$
\ALG_i\le\min\{5B,\tfrac54m_i g_i\}\le5B\min\{1,m_i2^{-k}\}.
$$
Put $L=\lceil\log_2(N+1)\rceil$.
Levels with $k\le L$ contribute at most $5B(L+1)$.
The remaining tail is at most
$$
5B\sum_{k>L}m_{i^\star-k}2^{-k}\le5B\,2^{-L}\sum_i m_i\le5B.
$$
Thus $\ALG\le5\bigl(\lceil\log_2(N+1)\rceil+2\bigr)\OPT$, giving an $O(\log(N+1))$ upper bound on the competitive ratio.

We next prove the matching lower bound of $\Omega(\log(N+1))$.
Fix $K\ge1$, set $g_i=2^i$ for $0\le i\le K$, let $T=g_K$, and put $\delta=1/(4T)$.
Initially release one request at each location $g_i$.
Whenever a service clears the current request at $g_i$ at time $t$, release its successor there at $t+\delta$ provided $t+\delta\le T$.
Determinism makes this a fixed input after simulation.

Every active request at $g_i$ is cleared within $g_i/4$, since otherwise level $g_i$ becomes tight and forces a service of length at least $g_i$.
If $C_i$ is the number of clearances at $g_i$ by time $T$, completed lifetimes, $\delta$-gaps, and a final residual interval of length at most $g_i/4+\delta$ cover $[0,T]$.
Hence
 \begin{equation}
 \label{eq:max-balance-scale-clearances}
C_i\ge\frac{T}{g_i/4+\delta}-1, \qquad g_iC_i\ge\frac{4T^2}{T+1}-g_i.
 \end{equation}

Let $s_j$ be the number of exact length-$g_j$ services by time $T$.
Their witness intervals are disjoint, so $s_j\le4T/g_j+1$.
Such a service clears at most one active request at each of $g_0,\ldots,g_j$.
Including initial requests,
 \begin{align}
N&\le K+1+\sum_{j=0}^K(j+1)s_j\\ &=O\left(K^2+T\sum_{j=0}^K\frac{j+1}{2^j}\right)=O(T),
 \label{eq:max-balance-n-upper}
 \end{align}
while \eqref{eq:max-balance-scale-clearances} at $i=0$ gives $N\ge C_0=\Omega(T)$.
Thus
 \begin{equation}
 \label{eq:max-balance-lower-request-count}
N=\Theta(T)=\Theta(2^K).
 \end{equation}

A length-$g_j$ service clears at most one current request at each $g_i\le g_j$, and $\sum_{i=0}^jg_i<2g_j$.
Therefore
$$
\sum_{\text{Max-Balance services}}\!\!\text{length}\ge\frac12\sum_{i=0}^K g_iC_i\ge\frac12\left((K+1)\frac{4T^2}{T+1}-\sum_{i=0}^Kg_i\right)=\Omega(KT).
$$
The offline optimum can wait until $T$ and issue one length-$T$ service, whose cost is at most $2T$.
Since $K$ is arbitrary, the construction gives arbitrarily large instances for which $\frac{\ALG}{\OPT}=\Omega(K)=\Omega(\log(N+1))$ by \eqref{eq:max-balance-lower-request-count}.
Together with the upper bound, this proves that the competitive ratio is $\Theta(\log(N+1))$.
\end{proof}

\bibliographystyle{plain}
\bibliography{references}

\end{document}